\documentclass[sigconf]{aamas}

\usepackage{balance} 

\usepackage{amsthm}
\usepackage{mathtools}
\usepackage{color}
\usepackage{amsfonts}
\newcommand{\bs}{\boldsymbol}

\newcommand{\st}{{\rm st}}

\newcommand{\rmd}{{\rm d}}
\newcommand{\sumi}{\Sigma_i}
\newcommand{\sumip}{\Sigma_{i'}}
\newcommand{\sumipp}{\Sigma_{i''}}
\newcommand{\sumj}{\Sigma_j}
\newcommand{\sumjp}{\Sigma_{j'}}
\newcommand{\sumk}{\Sigma_k}
\newtheorem{theorem}{Theorem}

\newtheorem{definition}{Definition}
\newtheorem{corollary}{Corollary}
\usepackage{comment}
\usepackage{bm}

\makeatletter
\gdef\@copyrightpermission{
  \begin{minipage}{0.2\columnwidth}
   \href{https://creativecommons.org/licenses/by/4.0/}{\includegraphics[width=0.90\textwidth]{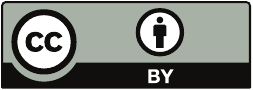}}
  \end{minipage}\hfill
  \begin{minipage}{0.8\columnwidth}
   \href{https://creativecommons.org/licenses/by/4.0/}{This work is licensed under a Creative Commons Attribution International 4.0 License.}
  \end{minipage}
  \vspace{5pt}
}
\makeatother

\setcopyright{ifaamas}
\acmConference[AAMAS '25]{Proc.\@ of the 24th International Conference
on Autonomous Agents and Multiagent Systems (AAMAS 2025)}{May 19 -- 23, 2025}
{Detroit, Michigan, USA}{Y.~Vorobeychik, S.~Das, A.~Nowé  (eds.)}
\copyrightyear{2025}
\acmYear{2025}
\acmDOI{}
\acmPrice{}
\acmISBN{}

\acmSubmissionID{662}

\title[Global Behavior of Learning Dynamics in Zero-Sum Games with Memory Asymmetry]{Global Behavior of Learning Dynamics in Zero-Sum Games\\ with Memory Asymmetry}

\author{Yuma Fujimoto}
\affiliation{
  \institution{CyberAgent}\country{}
  \institution{University of Tokyo}\country{}
  \institution{Soken University}\country{}
  }
\email{fujimoto.yuma1991@gmail.com}

\author{Kaito Ariu}
\affiliation{
  \institution{CyberAgent}\country{}
  }
\email{kaito_ariu@cyberagent.co.jp}

\author{Kenshi Abe}
\affiliation{
  \institution{CyberAgent}\country{}
  }
\email{abekenshi1224@gmail.com}

\begin{abstract}
This study examines the global behavior of dynamics in learning in games between two players, X and Y. We consider the simplest situation for memory asymmetry between two players: X memorizes the other Y's previous action and uses reactive strategies, while Y has no memory. Although this memory complicates their learning dynamics, we characterize the global behavior of such complex dynamics by discovering and analyzing two novel quantities. One is an extended Kullback-Leibler divergence from the Nash equilibrium, a well-known conserved quantity from previous studies. The other is a family of Lyapunov functions of X's reactive strategy. One of the global behaviors we capture is that if X exploits Y, then their strategies converge to the Nash equilibrium. Another is that if Y's strategy is out of equilibrium, then X becomes more exploitative with time. Consequently, we suggest global convergence to the Nash equilibrium from both aspects of theory and experiment. This study provides a novel characterization of the global behavior in learning in games through a couple of indicators.
\end{abstract}


\newcommand{\BibTeX}{\rm B\kern-.05em{\sc i\kern-.025em b}\kern-.08em\TeX}

\begin{document}


\pagestyle{fancy}
\fancyhead{}


\maketitle 


\section{Introduction}
Learning in games targets how multiple agents learn their optimal strategies in the repetition of games~\cite{fudenberg1998theory}. The set of such players' best strategies is defined as Nash equilibrium~\cite{nash1950equilibrium}, where every player has no motivation to change his/her strategy. However, this equilibrium is hard to compute in general because one's best strategy depends on the others' strategies. Indeed, the behavior of multi-agent learning is complicated in zero-sum games, where players conflict in their payoffs. Even when players try to learn their optimal strategies there, their strategies often cycle around the Nash equilibrium and fail to converge to the equilibrium.

In order to understand such strange behaviors, which are unique in multi-agent learning, the dynamics of how multiple agents learn their strategies, say, learning dynamics, are frequently studied~\cite{tuyls2006evolutionary, busoniu2008comprehensive, tuyls2012multiagent, bloembergen2015evolutionary}. The representative dynamics of interest are the replicator dynamics, which is based on the evolutionary dynamics in biology~\cite{taylor1978evolutionary, friedman1991evolutionary, borgers1997learning, hofbauer1998evolutionary, nowak2004evolutionary}. These dynamics are also known as the multiplicative weight updates (MWU) in its discrete-time version~\cite{arora2012multiplicative, bailey2018multiplicative}. Furthermore, their connection to other representative learning dynamics, such as gradient ascent~\cite{singh2000nash, bowling2002multiagent, zinkevich2003online, bowling2004convergence} and Q-learning~\cite{watkins1992q, hussain2023asymptotic, hussain2023beyond}, should be noted. Such replicator dynamics are known to be characterized by Kullbuck-Leibler (KL) divergence, which is the distance from the Nash equilibrium to the players' present strategies. This KL divergence is conserved during the learning dynamics, and the distance from the Nash equilibrium is invariant~\cite{piliouras2014optimization, piliouras2014persistent}. Follow the Regularized Leader (FTRL) is a class of learning algorithms including the replicator dynamics and also has its conserved quantity, which is the summation of divergences for all the players~\cite{mertikopoulos2016learning, mertikopoulos2018cycles}. To summarize, such complex learning dynamics have been discussed based on their conserved quantity.

In this study, we define memory as an agent's ability to change its action choice depending on the outcome of past games. By definition, this memory allows the agent to make more complex and intelligent decisions. When memory is introduced into a normal-form game, the players can achieve a wider range of strategies as the Nash equilibria (known as Folk theorem~\cite{fudenberg2009folk}). Furthermore, memory is also introduced into learning algorithms, such as replicator dynamics~\cite{fujimoto2019emergence, fujimoto2021exploitation, fujimoto2023learning, fujimoto2024memory} and Q-learning~\cite{masuda2009theoretical, barfuss2020reinforcement, usui2021symmetric, meylahn2022limiting, ueda2023memory}. Here, since this memory causes feedback from the past, the global dynamics of such learning algorithms become more complex. Indeed, replicator dynamics diverge from the Nash equilibrium under symmetric memory lengths between players~\cite{fujimoto2023learning}, while converging under asymmetric memory lengths~\cite{fujimoto2024memory}. Here, KL divergence is no longer useful to capture the global dynamics because it increases or decreases over time. The analysis of the dynamics in with-memory games is limited to the local, linearized stability analysis in the vicinity of Nash equilibrium~\cite{fujimoto2024memory}. To summarize, since memory crucially complicates learning dynamics, the global behavior of the dynamics is still unexplored.

\begin{figure*}[tb]
    \centering
    \includegraphics[width=0.6\hsize]{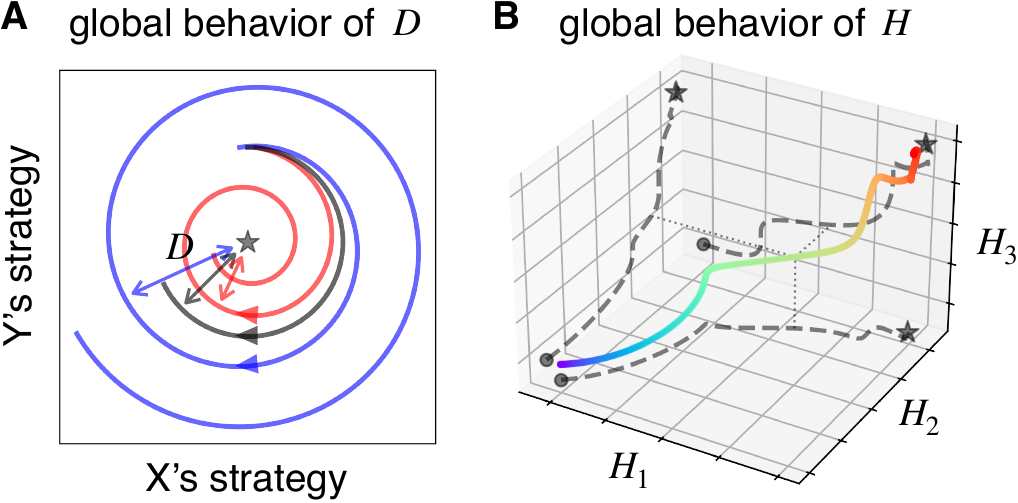}
    \caption{(A) Illustration of the global behavior of the conditional divergence, $D(\bs{X},\bs{y})$. Three trajectories (red, black, and blue) are plotted with the Nash equilibrium (the black star marker). The horizontal and vertical axes show X's strategy ($x_{1}^{\st}$) and Y's strategy ($y_{1}$) in the matching pennies game (formulated in Fig.~\ref{F02}). This divergence decreases (red: $\dot{D}<0$), cycles (black: $\dot{D}=0$), or increases (blue: $\dot{D}>0$) with time. These three lines are plotted for the different initial strategies, i.e., $\bs{X}$ and $\bs{y}$. (B) Illustration of the global behavior of the family of Lyapunov functions, $H(\bs{X};\bs{\delta})$. The colored line shows a trajectory (from purple to red) of Lyapunov functions $H_1$, $H_2$, and $H_3$, each of which is $H(\bs{X};\bs{\delta})$ for some specific $\bs{\delta}$. The gray broken lines are the projections of the black solid line to $H_1$-$H_2$, $H_2$-$H_3$, and $H_3$-$H_1$ planes. All of $H_1$, $H_2$, and $H_3$ monotonically increase with time.}
    \label{F01}
\end{figure*}

This study provides the first theoretical analysis of the global behavior of learning in with-memory games. We assume games where their memory structure is simplest and asymmetric; One side adopts a reactive strategy that can memorize the other's previous action~\cite{nowak1990stochastic, ohtsuki2004reactive, baek2016comparing, fujimoto2019emergence, schmid2022direct, fujimoto2021exploitation}, while the other has no memory. In order to characterize the global behavior of such with-memory games, we extend KL divergence and prove that such extended divergence increases or decreases with time depending on whether the reactive strategy is exploitative or not (see Fig.~\ref{F01}A). We further propose a family of Lyapunov functions that characterize the dynamics of the reactive strategy (see Fig.~\ref{F01}B). These Lyapunov functions show that the with-memory side monotonically learns to exploit the no-memory side. As an application of these functions, we suggest the convergence from arbitrary initial strategies to the equilibrium, i.e., global convergence. We prove global convergence in the matching pennies game. We also experimentally confirm that such global convergence is observed in other games equipped with various types of equilibrium.


\section{Preliminary}
\subsection{Settings}
First, we formulate two-player normal-form games. We consider two players, denoted as X and Y. X's actions are denoted as $\{a_i\}_{1\le i\le m_{\sf X}}$, while Y's are $\{b_j\}_{1\le j\le m_{\sf Y}}$. When X and Y choose $a_i$ and $b_j$, they obtain the payoffs of $u_{ij}\in\mathbb{R}$ and $v_{ij}\in\mathbb{R}$, respectively. Thus, all their possible payoffs are given by the matrices, $\bs{U}:=(u_{ij})_{ij}$ and $\bs{V}:=(v_{ij})_{ij}$. Here, when $\bs{V}=-\bs{U}$ holds, the games are called zero-sum. Although our formulation of learning algorithms can be used for general games, this study focuses on zero-sum games.

\begin{figure*}[h!]
    \centering
    \includegraphics[width=0.55\hsize]{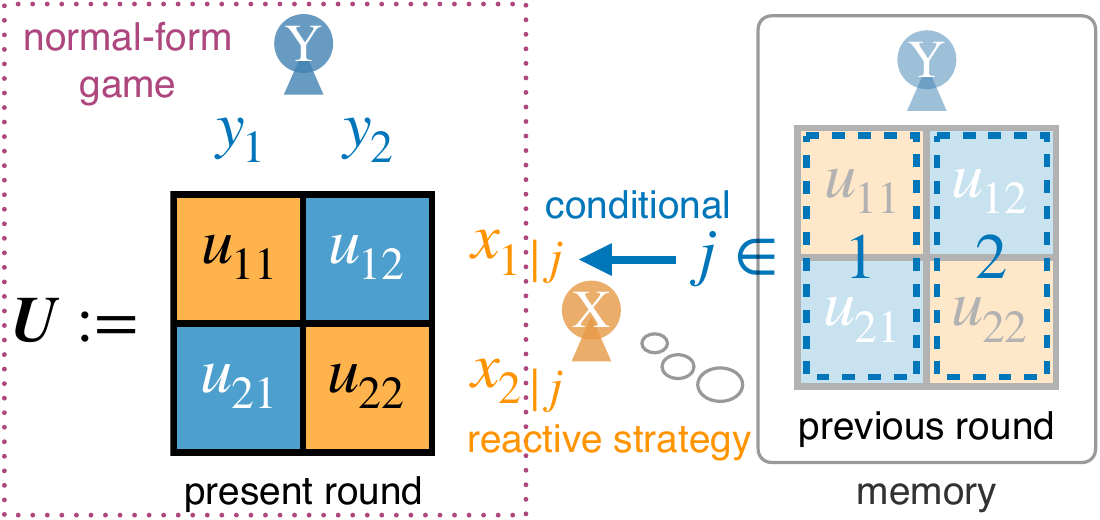}
    \caption{Illustration of games between reactive and zero-memory strategies. The area surrounded by the magenta dotted line shows the normal-form game. In each round, X chooses action $i=1$ or $2$ in the row, following its strategy, i.e., the probability distribution of $\bs{x}=(x_1,x_2)$. On the other hand, Y chooses action $j=1$ or $2$ in the column, following its strategy, i.e., the probability distribution of $\bs{y}=(y_1,y_2)$. Depending on their choices $i$ and $j$, X receives a payoff $u_{ij}$, given by a matrix form of $\bs{U}=(u_{ij})_{i,j}=((u_{11},u_{12}), (u_{21},u_{22}))$. Furthermore, in zero-sum games, Y receives $-u_{ij}$. Especially in the matching pennies game, their actions of $1$ ($2$) correspond to the choice of ``head'' (``tail'') of a coin. When their choices match $i=j$, X wins, i.e., $u_{11}=u_{22}=1$ (the orange blocks). Else when their choices mismatch $i\neq j$, Y wins, i.e., $u_{12}=u_{21}=-1$ (the blue blocks). The area outside of the magenta dotted line shows the difference due to an effect of memory. The gray box shows that X memorizes Y's previous action, represented as $j=1$ or $2$. Thus, X uses a reactive strategy and can choose its action with the conditional probability of $x_{1|j}$ and $x_{2|j}$ for Y's previous action.}
    \label{F02}
\end{figure*}

We assume that X use reactive strategies, i.e., can change its action choice depending on the other's previous action. This reactive strategy is denoted as $\bs{X}:=(x_{i|j})_{1\le i\le m_{\sf X}, 1\le j\le m_{\sf Y}}\in \prod_{1\le j\le m_{\sf Y}}\Delta^{m_{\sf X}-1}$, a matrix composed of $m_{\sf Y}$ vectors each of which are an element of a $m_{\sf X}-1$-dimensional simplex. Here, $x_{i|j}$ means the probability that X chooses $a_{i}$ in the condition when Y's previous action is $b_{j}$. Thus, $\sumi x_{i|j}=1$ should be satisfied for all $j$. On the other hand, Y only can use classical mixed strategies and choose its own action without reference to the previous actions. This mixed strategy is denoted as $\bs{y}=(y_{j})_{1\le j\le m_{\sf Y}}\in \Delta^{m_{\sf Y}-1}$, a vector which is an element of a $(m_{\sf Y}-1)$-dimensional simplex. Thus, $\sumj y_{j}=1$ should be satisfied.

\subsection{Stationary State and Expected Payoff}
We now discuss the stationary state and expected payoff of repeated games. Since Y determines its action independent of the outcomes of previous rounds, X's stationary strategy, defined as $\bs{x}^{\st}:=(x_{i}^{\st})_{1\le i\le m_{\sf X}}$, is given by $x_{i}^{\st}(\bs{x}_{i},\bs{y})=\sumj x_{i|j}y_{j}$.
Here, $x_{i}^{\st}$ means the probability that X chooses $a_i$ in the stationary state. Furthermore, the stationary state is described as $\bs{P}^{\st}:=\bs{x}^{\st}\otimes \bs{y}$ with use of $\bs{x}^{\st}$ and $\bs{y}$. Last, X's expected payoff is given by $u^{\st}(\bs{x}^{\st},\bs{y}):=\sumi\sumj u_{ij}p_{ij}^{\st}=\sumi\sumj u_{ij}x_{i}^{\st}y_{j}$.

\subsection{Nash Equilibrium}
We now define the Nash equilibrium in the normal-form game. Here, note that this equilibrium is based on games without memories, where X's strategy does not refer to the past games, i.e., $\bs{x}:=(x_{i})_{i}$. By using the expected payoff $u^{\st}(\bs{x},\bs{y})$ for games without memories, the Nash equilibrium $(\bs{x}^{*},\bs{y}^{*})$ is formulated as
\begin{align}
    \begin{cases}
        \bs{x}^{*}={\rm argmax}_{\bs{x}}u^{\st}(\bs{x},\bs{y}^*) \\
        \bs{y}^{*}={\rm argmin}_{\bs{y}}u^{\st}(\bs{x}^*,\bs{y}) \\
    \end{cases}.
\end{align}
From the definition, $u^{\st}(\bs{x},\bs{y})$ is the linear function for $\bs{x}$ and $\bs{y}$, and the Nash equilibrium condition is characterized by the gradient of such expected payoffs as
\begin{align}
    &\begin{cases}
        \partial u^{\st}/\partial x_{i}=\sumj u_{ij}y_{j}^{*}=C & (x_{i}^{*}>0) \\
        \partial u^{\st}/\partial x_{i}=\sumj u_{ij}y_{j}^{*}\le C & (x_{i}^{*}=0) \\
    \end{cases}, \\
    &\begin{cases}
        \partial u^{\st}/\partial y_{j}=\sumi u_{ij}x_{i}^{*}=C & (y_{j}^{*}>0) \\
        \partial u^{\st}/\partial y_{j}=\sumi u_{ij}x_{i}^{*}\ge C & (y_{j}^{*}=0) \\
    \end{cases}.
    \label{Nash_condition}
\end{align}
From these conditions, for all $i$ and $j$ such that $x_{i}^{*}>0$ and $y_{j}^{*}>0$ hold, respectively, we obtain
\begin{align}
    \sumi u_{ij}x_{i}^{*}=\sumj u_{ij}y_{j}^{*}=:u^{*}.
    \label{same_payoff}
\end{align}
Let us interpret this equation. First, $\sumi u_{ij}x_{i}^{*}=u^{*}$ means that when X takes its Nash equilibrium strategy, its own payoff is fixed to $u^{*}$, independent of Y's strategy. On the other hand, $\sumj u_{ij}y_{j}^{*}=u^{*}$ similarly means that Y's Nash equilibrium strategy fixes X's payoff to $u^{*}$. In other words, either X or Y takes its Nash equilibrium strategy, their payoffs are fixed. This is the special property in zero-sum games.

\subsection{Learning Algorithm: Replicator Dynamics}
Let us define the replicator dynamics as a representative learning algorithm. X's and Y's replicator dynamics are formulated as
\begin{align}
    \dot{x}_{i|j}&=+x_{i|j}\left(\frac{\rmd u^{\st}}{\rmd x_{i|j}}-\sumi x_{i|j}\frac{\rmd u^{\st}}{\rmd x_{i|j}}\right),
    \label{dotx} \\
    \dot{y}_{j}&=-y_{j}\left(\frac{\rmd u^{\st}}{\rmd y_{j}}-\sumj y_j\frac{\rmd u^{\st}}{\rmd y_{j}}\right).
    \label{doty}
\end{align}
Here, following the theorems in~\cite{fujimoto2023learning}, X's and Y's replicator dynamics include the gradient for the expected payoff $u^{\st}$. Thus, the update of X's strategy increases its payoff $u^{\st}$, while that of Y's strategy decreases the other's payoff $u^{\st}$. We discuss learning based on the replicator dynamics throughout this study, but we can extend all the following results to another typical learning algorithm, the gradient descent-ascent (see Appendix~\ref{app_extension} for detailed discussion).

\section{Theory on Learning Dynamics}
This section analyzes the dynamics of Eqs.~\eqref{dotx} and~\eqref{doty}. First, we compute in detail the gradient terms, which appear to be complex. Next, as a preliminary, we define positive definite matrices for some special vectors. Based on this definition, we introduce two quantities characterizing the dynamics of Eqs.~\eqref{dotx} and~\eqref{doty}: An extended KL divergence and a family of Lyapunov functions.

\subsection{Polynomial Expressions of Learning}
First, the gradient terms in Eqs.~\eqref{dotx} and~\eqref{doty} are computed as
\begin{align}
    \frac{\rmd u^{\st}(\bs{x}^{\st}(\bs{X},\bs{y}),\bs{y})}{\rmd x_{i|j}}&=\frac{\partial x_{i}^{\st}(\bs{x}_{i},\bs{y})}{\partial x_{i|j}}\frac{\partial u^{\st}(\bs{x}^{\st},\bs{y})}{\partial x_{i}^{\st}} \\
    &=y_{j}\sumjp u_{ij'}y_{j'},
    \label{gradx} \\
    \frac{\rmd u^{\st}(\bs{x}^{\st}(\bs{X},\bs{y}),\bs{y})}{\rmd y_{j}}&=\frac{\partial u^{\st}(\bs{x}^{\st},\bs{y})}{\partial y_{j}}+\sumi \frac{\partial x_{i}^{\st}(\bs{x}_{i},\bs{y})}{\partial y_{j}}\frac{\partial u^{\st}(\bs{x}^{\st},\bs{y})}{\partial x_{i}^{\st}} \\
    &=\sumi u_{ij}x_{i}^{\st}+\sumi x_{i|j}\sumjp u_{ij'}y_{j'}.
    \label{grady}
\end{align}
Here, we remark that Eqs.~\eqref{dotx} and~\eqref{doty} are nonlinear functions of $\bs{X}$ and $\bs{y}$, which is a feature of learning in with-memory games. Notably, however, these equations are polynomial expressions with $\bs{X}$ and $\bs{y}$. Such polynomial expressions cannot be seen in the games of other memory lengths~\cite{fujimoto2023learning, fujimoto2024memory} but are special in the games between reactive and no memory strategies.

\subsection{Positive Definiteness for Zero-Sum Vectors}
Next, let us introduce a definiteness of matrices. Here, however, this definite matrix is for vectors whose elements are summed to $0$, named ``zero-sum vectors''. In mathematics, zero-sum vector $\bs{\delta}:=(\delta_{k})_{k}$ satisfies $\sumk\delta_{k}=0$ but $\bs{\delta}\neq \bs{0}$.

\begin{definition}[Positive definiteness for zero-sum vectors]
A square matrix $\bs{M}$ is ``positive definite for zero-sum vectors'' when for all vectors $\bs{\delta}\neq \bs{0}$ such that $\sumk \delta_{k}=0$, $\bs{\delta}\cdot(\bs{M}\bs{\delta})<0$ holds.
\end{definition}

The positive definiteness for zero-sum vectors connects with an ordinary positive definiteness by a simple transformation of a matrix (see Appendix~\ref{app_definite} for details).

\subsection{Extended Kullback-Leibler Divergence}
The first quantity is an extended version of divergence. Before considering the extension, we introduce the classical version of divergence $D_{\rm c}$, which is the function of X's mixed strategies ($\bs{x}:=(x_{i})_{1\le i\le m_{\sf X}}\in\Delta^{m_{\sf X}-1}$) and Y's mixed strategies ($\bs{y}$) as
\begin{align}
    D_{\rm c}(\bs{x},\bs{y})&:=D_{\rm KL}(\bs{x}^{*}\|\bs{x})+D_{\rm KL}(\bs{y}^{*}\|\bs{y}), \\
    D_{\rm KL}(\bs{p}^*\|\bs{p})&:=\bs{p}^{*}\cdot\log \bs{p}^{*}-\bs{p}^{*}\cdot\log \bs{p}.
\end{align}
We now give an intuitive interpretation of this quantity. First, $D_{\rm KL}(\bs{p}^*\|\bs{p})$ is the KL divergence, meaning the distance from the reference point $\bs{p}^*$ to the target point $\bs{p}$. Thus, $D_{\rm c}(\bs{x},\bs{y})$ means the total distance from the Nash equilibrium $(\bs{x}^{*},\bs{y}^{*})$ to the current state $(\bs{x},\bs{y})$.

Let us extend the classical divergence to the case of this study, where X refers to the previous action of the other and can use reactive strategies $\bs{X}$. This extended divergence, i.e., $D(\bs{X},\bs{y})$, is named the ``conditional-sum'' divergence, formulated as
\begin{align}
    D(\bs{X},\bs{y})&:=\sumj D_{\rm KL}(\bs{x}^{*}\|\bs{x}_{j})+D_{\rm KL}(\bs{y}^{*}\|\bs{y}).
    \label{D_function}
\end{align}
We now remark the difference between $D(\bs{X},\bs{y})$ and $D_{\rm c}(\bs{x},\bs{y})$. Recall that X's reactive strategy is defined as $(\bs{x}_{j})_{1\le j\le m_{\sf Y}}$, which shows how to choose its action with the condition that Y chose $b_{j}$ in the previous round. Hence, $D(\bs{X},\bs{y})$ represents the summation of KL divergence from $\bs{x}^{*}$ to $\bs{x}_{j}$ for all the conditions of $j$. Here, we also remark that when the reactive strategy does not use memory, i.e., $\bs{x}_{j}=\bs{x}$ for all $j$, this conditional-sum divergence also captures the behavior of the classical divergence (see Appendix~\ref{app_divergence} for details).

This conditional-sum divergence satisfies the following theorem (see Appendix~\ref{proof_thm_D} for its full proof).

\begin{theorem}[Monotonic decrease of $D$ for positive definite $\bs{X}^{\rm T}\bs{U}$]
\label{thm_D}
If $\bs{X}^{\rm T}\bs{U}$ is positive definite for zero-sum vector, $D^{\dagger}(\bs{X};\rmd\bs{y}):=\dot{D}(\bs{X},\bs{y})<0$ for all $\rmd\bs{y}:=\bs{y}-\bs{y}^{*}\neq \bs{0}$.
\end{theorem}

\textsc{Proof Sketch.} We calculation $\dot{D}(\bs{X},\bs{y})$ in practice. In the calculation, the contribution of X's gradient (Eq.~\eqref{gradx}) cancels out the contribution of the first term of the gradient of Y (Eq.~\eqref{grady}). (Here, we remark that the same canceling out also occurs in the calculation for the conservation of the classical divergence $D_{\rm c}(\bs{x},\bs{y})$ in games without memory.) However, the contribution of the second term of Eq.~\eqref{grady} is special in games of a reactive strategy. By using the constant payoff condition in the Nash equilibrium (Eqs.~\eqref{same_payoff}), we obtain
\begin{align}
    \dot{D}(\bs{X},\bs{y})=-\rmd\bs{y}^{\rm T}\bs{X}^{\rm T}\bs{U}\rmd\bs{y}\ (=:D^{\dagger}(\bs{X};\rmd\bs{y})),
\end{align}
which means the difference from Y's equilibrium strategy and is a zero-sum vector. Thus, when $\bs{X}^{\rm T}\bs{U}$ is positive definite for zero-sum vectors, $\rmd\bs{y}^{\rm T}\bs{X}^{\rm T}\bs{U}\rmd\bs{y}$ is always positive, leading to $D^{\dagger}(\bs{X};\rmd\bs{y})<0$ for all $\rmd\bs{y}\neq \bs{0}$. \qed \\

\subsection{Family of Lyapunov Functions}
Furthermore, we introduce a Lyapunov function, which characterizes the learning dynamics of X's reactive strategy. Based on an arbitrary zero-sum vector $\bs{\delta}:=(\delta_{i})_{1\le i\le m_{\sf X}}$, this function is defined as
\begin{align}
    H(\bs{X};\bs{\delta}):=\bs{\delta}^{\rm T}\bs{U}\log\bs{X}^{\rm T}\bs{\delta}.
    \label{H_function}
\end{align}
The following theorem holds for this Lyapunov function (see Appendix~\ref{proof_thm_H} for its full proof).

\begin{theorem}[Monotonic increase of $H$]
\label{thm_H}
For all $\bs{\delta}$ such that $\sumi \delta_{i}=0$, $H^{\dagger}(\bs{y};\bs{\delta}):=\dot{H}(\bs{X};\bs{\delta})\ge 0$. The equality holds if and only if $d\bs{y}(=\bs{y}-\bs{y}^{*})=\bs{0}$.
\end{theorem}

\textsc{Proof Sketch.} By using Eq.~\eqref{same_payoff}, we calculate
\begin{align}
    \dot{H}(\bs{X};\bs{\delta})=|\bs{\delta}^{\rm T}\bs{U}\rmd\bs{y}|^2\ (=:H^{\dagger}(\bs{y};\bs{\delta})).
\end{align}
This means that $H^{\dagger}(\bs{y};\bs{\delta})\ge 0$ for all $\bs{\delta}$. If we substitute $\bs{\delta}=\bs{x}-\bs{x}^{*}$ for some $\bs{x}\in\Delta^{m_{\sf X}-1}$, we obtain
\begin{align}
    H^{\dagger}(\bs{y};\bs{x}-\bs{x}^{*})=|u^{\st}(\bs{x},\bs{y})-u^{*}|^2.
\end{align}
For any $\bs{y}\neq\bs{y}^{*}\Leftrightarrow\rmd\bs{y}\neq\bs{0}$, there is $\bs{x}$ such that $u^{\st}(\bs{x},\bs{y})-u^{*}\neq 0$. Thus, $\rmd\bs{y}=\bs{0}$ is equivalent to $H^{\dagger}(\bs{y};\bs{\delta})=0$ for all $\bs{\delta}$. \qed \\

Let us interpret the function of $H(\bs{X};\bs{\delta})$. First, $H(\bs{X};\bs{\delta})$ is the quadratic form of matrix $\bs{U}\log\bs{X}^{\rm T}$ for $\bs{\delta}$. We now focus on the meaning of $\bs{U}\log\bs{X}^{\rm T}$. The $i'$, $i$ element of $\bs{U}\log\bs{X}^{\rm T}$ is given by $\bs{u}_{i'}\cdot\log\bs{x}_{i}$, in which we denoted vector $\bs{x}_{i}:=(x_{i|j})_{j}$. This is the inner product of the payoff under X taking $i'$-th action ($\bs{u}_{i'}$) and the logarithmic strategy under X taking $i$-th action ($\log\bs{x}_{i}$). Thus, $\bs{U}\log\bs{X}^{\rm T}$ shows a correspondence matrix between X's strategy and its payoff matrix. Since $H(\bs{X};\bs{\delta})$ is the quadratic form of the correspondence matrix, it tends to be large when the diagonal elements of the correspondence matrix are large. The diagonal elements, i.e., $\bs{u}_{i}\cdot\log\bs{x}_{i}$, indicate how X exploits Y for each $i$-th action. Thus, $H(\bs{X};\bs{\delta})$ means the degree of exploitation of Y by X.

\begin{figure*}[tb]
    \centering
    \includegraphics[width=0.65\hsize]{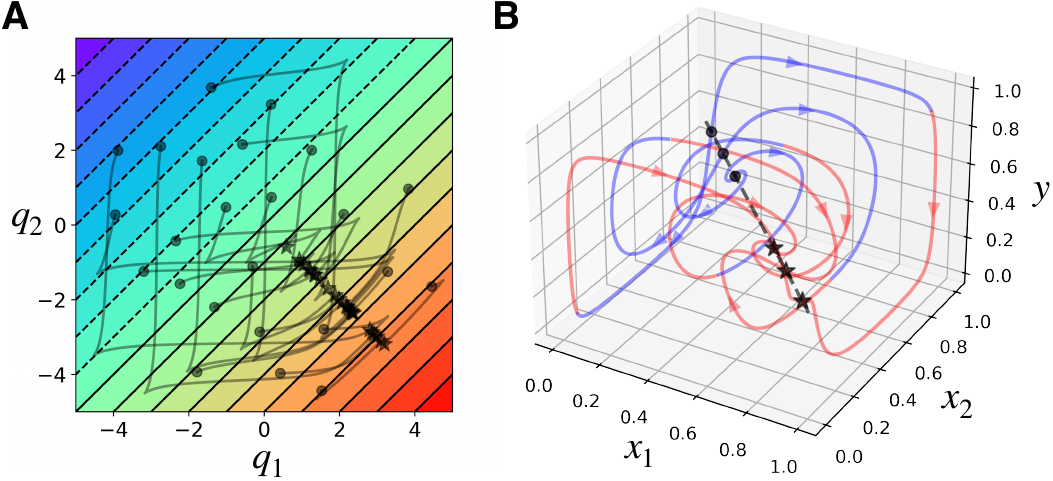}
    \caption{(A) Trajectories of $q_1$ and $q_2$. The rainbow contour plot indicates the value of $q_1-q_2$. All the trajectories monotonically increase $q_1-q_2$ with time and converge in the area of $q_1>q_2$ in their final states. (B) Trajectories of the learning dynamics. The black broken line corresponds to the region of Nash equilibria, $\bs{x}^{\st}=\bs{y}=(1/2,1/2)$. Each colored line shows a trajectory of the learning dynamics. First, the circle markers show the initial states. Following the blue lines, the trajectories diverge from the Nash equilibria ($D(\bs{X},\bs{y})$ increases with time). However, the trajectories stop to diverge and switch to converge to the Nash equilibria ($D(\bs{X},\bs{y})$ decreases), following the red lines. The star markers are the final states and correspond to one of the Nash equilibria.}
    \label{F03}
\end{figure*}

\subsection{Global Behavior by Two Quantities}
So far, the two theorems (Thms.~\ref{thm_D} and~\ref{thm_H}) explain how the two quantities vary with time. Let us understand the global behavior of the learning dynamics by interpreting the two quantities.

\paragraph{$D$ explains increasing/decreasing of distance:} First, recall that $D(\bs{X},\bs{y})$ means the distance from the Nash equilibrium. Thm.~\ref{thm_D} shows that the distance becomes smaller when $\bs{X}^{\rm T}\bs{U}$ is positive definite, whereas larger when $\bs{X}^{\rm T}\bs{U}$ is negative definite. Here, we focus on $\bs{X}^{\rm T}\bs{U}$ which determines whether the distance becomes smaller or larger. The $j$, $j'$ element of $\bs{X}^{\rm T}\bs{U}$ is given by $\bs{x}_{j}\cdot\bs{u}_{j'}$, in words, the correspondence between X's strategy for $j$-th action ($\bs{x}_{j}$) and its payoff for $j'$-th action ($\bs{u}_{j'}$). The eigenvalues of a matrix are roughly determined by the diagonal elements $\bs{x}_{j}\cdot\bs{u}_{j}$. This diagonal element is larger when $x_{i|j}$ takes a larger value for larger $u_{ij}$, meaning that X exploits Y's payoff more. As the simplest example, $\bs{X}^{\rm T}\bs{U}$ is positive definite when X takes $x_{i|j}=1$ for $i=\hat{i}$ such that $\hat{i}=\arg\max_{i}u_{ij}$, while $x_{i|j}=0$ for $i\neq\hat{i}$. To summarize, Thm.~\ref{thm_D} captures the global behavior where if X exploits Y, their strategies converge to the Nash equilibrium; otherwise, they diverge.

\paragraph{$H$ explains the monotonic increase of exploitability:} 
Next, recall that $\bs{U}\log\bs{X}^{\rm T}$ in $H(\bs{X};\bs{\delta})$ indicates a correspondence matrix between X's strategy and its payoff matrix. Thus, Thm.~\ref{thm_H} explains that unless Y takes the equilibrium strategy, this correspondence continues to increase with time. We also pay attention to the change in the degree of correspondence, i.e., $H^{\dagger}(\bs{y};\bs{\delta})$. By substituting $\bs{\delta}=\bs{x}-\bs{x}^{*}$ for some $\bs{x}\in\Delta^{m_{\sf X}-1}$, we rewrite it as $H^{\dagger}(\bs{y};\bs{x}-\bs{x}^{*})=|u^{\st}(\bs{x},\bs{y})-u^{*}|^2$. Here, $|u^{\st}(\bs{x},\bs{y})-u^{*}|$ indicates the difference in payoff from the equilibrium, i.e., the exploitation by X to Y. Therefore, the correspondence between X's strategy and its payoff matrix becomes larger according to the exploitability.



\paragraph{Remark:} We have interpreted both $\bs{X}^{\rm T}\bs{U}$ and $\bs{U}\log\bs{X}^{\rm T}$ as the degree of correspondence between X's strategy and its payoff matrix. The interpretation is qualitatively true, but we remark that there are several quantitative differences between $\bs{X}^{\rm T}\bs{U}$ and $\bs{U}\log\bs{X}^{\rm T}$, such as the order of multiplication and the existence of logarithm.

\section{Application: Global Convergence}
By combining the global behaviors obtained from Thms.~\ref{thm_D} and ~\ref{thm_H}, we expect that the global convergence to the Nash equilibrium occurs regardless of X's and Y's initial strategies. Thm.~\ref{thm_H} shows that as long as Y's strategy is out of equilibrium, the correspondence between X's strategy and its payoff matrix continues to be stronger. Afterward, Thm.~\ref{thm_D} shows that if the correspondence is sufficiently strong, Y's strategy is induced to the Nash equilibrium. In the following, our theory and experiment support that such global convergence occurs.

\subsection{Example: Matching Pennies}
Let us define the matching pennies game (see Fig.~\ref{F02} for the illustration of its payoff matrix). This game considers the action numbers of $m_{\sf X}=m_{\sf Y}=2$ and the payoff matrix of $\bs{U}=((u_{11},u_{12}), (u_{21},u_{22}))=((+1,-1), (-1,+1))$. The Nash equilibrium of this game is only $\bs{x}^{*}=\bs{y}^{*}=(1/2,1/2)$. This game is the simplest example of a game equipped with a full-support Nash equilibrium. In addition, it has been known that the replicator dynamics in games without memories cycle around the Nash equilibrium and cannot reach the equilibrium. Nevertheless, by considering the memory asymmetry (Eqs.~\eqref{dotx} and~\eqref{doty}), Y's strategy succeeds in the convergence to the equilibrium, as shown in the following corollary (see Appendix~\ref{proof_cor_global} for its full proof). For convenience, we use a special notation available in two-action games; $(x_{1|j},x_{2|j})=:(x_{j},1-x_{j})$, $q_{j}:=\log x_{j}-\log (1-x_{j})$, and $(y_{1},y_{2})=:(y,1-y)$.

\begin{corollary}[Global convergence in matching pennies]
\label{cor_global}
In the matching pennies game $\bs{U}=((+1,-1),(-1,+1))$, Y's strategy $\bs{y}$ converges to the equilibrium $\bs{y}^*$, independent of both the players' initial strategies.
\end{corollary}

\textsc{Proof Sketch.} First, note that $q_1>q_2\Leftrightarrow x_1>x_2$. By the direct calculation, we prove $H(\bs{X};\bs{\delta})\propto q_1-q_2$. Thm.~\ref{thm_H} shows that as long as $\bs{y}=\bs{y}^*$, $H(\bs{X};\bs{\delta})$ continues to increase. Thus, after a sufficiently long time, $H(\bs{X};\bs{\delta})>0\Leftrightarrow q_1>q_2\Leftrightarrow x_1>x_2$ continue to hold. We can also prove that $x_1>x_2$ is equivalent to the positive definiteness of $\bs{X}^{\rm T}\bs{U}$. Thm.~\ref{thm_D} shows that under positive definite $\bs{X}^{\rm T}\bs{U}$, $\bs{y}$ asymptotically converges to $\bs{y}^{*}$, its equilibrium strategy. \qed \\

Our experiments visualize the mechanism of the global convergence based on Thm.~\ref{thm_D} and~\ref{thm_H}. Fig.~\ref{F03}A shows the dynamics of $q_1$ and $q_2$. Here, the colors indicate the contour plot for $q_1-q_2$, showing that $H(\bs{X};\bs{\delta})\propto q_1-q_2$ monotonically increases with time and thus Thm.~\ref{thm_H} holds. Furthermore, we also see that X's strategy reaches the region of $q_1>q_2\Leftrightarrow x_1>x_2$ after a sufficiently long time passes. In the region, $\bs{X}^{\rm T}\bs{U}$ is positive definite, and thus Thm.~\ref{thm_D} is applicable after sufficiently long time passes.


\begin{figure*}[tb]
    \centering
    \includegraphics[width=0.9\hsize]{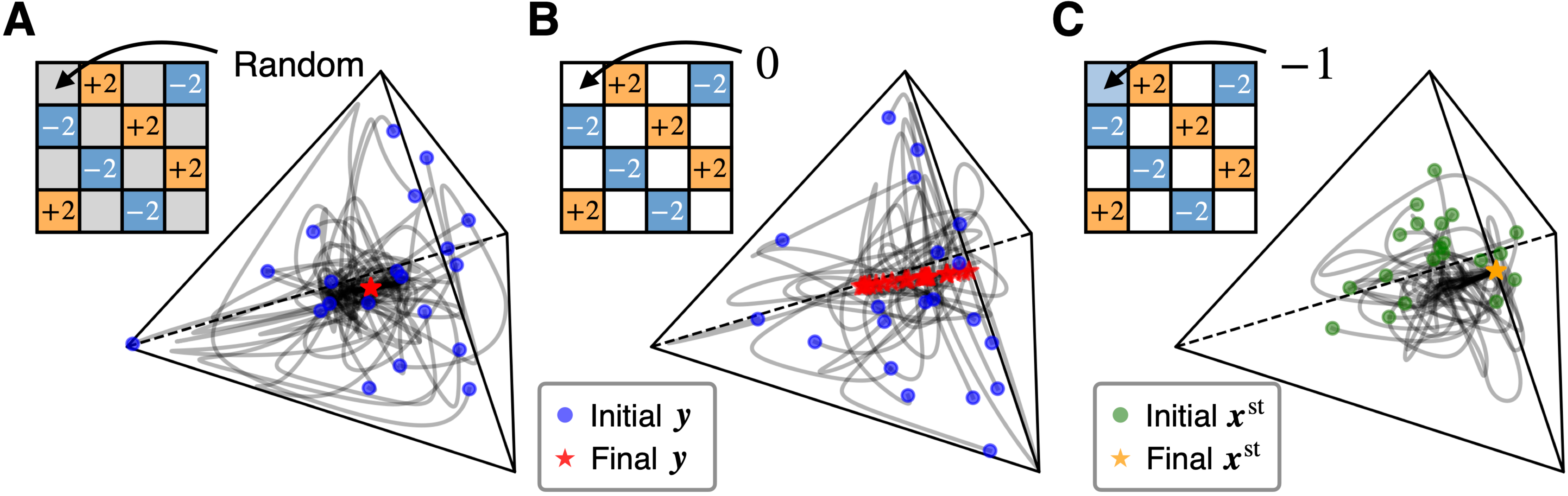}
    \caption{Global convergence in the coupled matching pennies games, where the second, third, fourth, and first actions win the other's first, second, third, and fourth actions, respectively. The winner receives the payoff of $2$ (the orange blocks in the matrices for the winning of X), while the loser sends the payoff of $2$ (the blue blocks). We now introduce three variants for the other blocks in the payoff matrix. (A) The case of interior equilibrium. We set each of the other blocks by random numbers in $[-1,1]$ (the gray blocks). Then, Y's strategy converges to the unique Nash equilibrium (the red star marker) independent of its initial state (the blue circle markers). (B) The case of continuous equilibrium. We set each of the other blocks by $0$, where the payoff matrix degenerates. Y's strategy converges to one of the Nash equilibria (the line consisting of the red star markers) depending on its initial state. (C) The case of boundary equilibrium. Only the block for the interaction between an action is set to $-1$, and the others are $0$. If so, X's strategy converges to the unique Nash equilibrium (the orange star markers) independent of its initial state (the green circle markers). Instead, Y's strategies do not converge.}
    \label{F04}
\end{figure*}

Next, Fig.~\ref{F03}B plots the global behavior of the learning dynamics, which is described by the three parameters of $x_1$, $x_2$, and $y$. The gray line shows the region that corresponds to the Nash equilibrium, i.e., $\bs{x}^{\st}(\bs{X},\bs{y})=\bs{y}=(1/2,1/2)$. Furthermore, the colored lines show example trajectories of the learning dynamics. The blue part of the line shows that $D(\bs{X},\bs{y})$ increases at the beginning of the learning dynamics. This part is in the region of $x_1<x_2$, following Thm.~\ref{thm_D}. After that, the red part shows that $D(\bs{X},\bs{y})$ decreases, and the learning dynamics converge to the equilibrium. This part is in the region of $x_1>x_2$, following Thm.~\ref{thm_D}.

\subsection{Example: Coupled Matching Pennies}
We also observe the global convergence in other zero-sum games beyond the matching pennies game. Fig.~\ref{F04} considers three examples of ``coupled'' matching pennies games, where a pair of matching pennies games are coupled with some interaction. The game considers the action numbers of $m_{\sf X}=m_{\sf Y}=4$, and some elements of the payoff matrix is fixed as $u_{ij}=+2$ for $j=\sigma(i)$ and $u_{ij}=-2$ for $i=\sigma(j)$. Here, we used the permutation function $\sigma$ as $\sigma(1)=2$, $\sigma(2)=3$, $\sigma(3)=4$, and $\sigma(4)=1$. If we consider only X's odd actions and Y's even actions or its reverse, the payoff matrix corresponds to the matching pennies game. Interestingly, there are various types of Nash equilibrium depending on the interaction between these matching pennies games, i.e., the other elements of the payoff matrix, $u_{ij}$ for neither $j=\sigma(i)$ nor $i=\sigma(j)$. Indeed, Fig.~\ref{F04} shows three cases where Nash equilibrium exists (A) in the interior, (B) continuously, and (C) on the boundary.

We remark that the global behavior obtained from Thm.~\ref{thm_D} and~\ref{thm_H} are available even though the trajectories of the learning dynamics look complicated (see Fig.~\ref{F04}). Furthermore, Panels A and B show that Y's strategy converges to the equilibrium. The only exception is Panel C, but X's strategy converges to the equilibrium instead of Y's. In the following, we explain in detail this convergence for each panel.

\paragraph{Interior equilibrium:} First, Fig.~\ref{F04}A shows the case where the other elements of the payoff matrix are random numbers following the uniform distribution of $[-1,1]$. In this case, the payoff matrix is linearly independent. In mathematics, there exists no $\bs{a}=(a_{j})_{1\le j\le m_{\sf Y}}\in \mathbb{R}^{m_{\sf Y}}$ such that $\sumj a_{j}\bs{u}_{j}=\bs{0}$ other than $\bs{a}=\bs{0}$. Thus, there is a single Nash equilibrium in the interior of the strategy space. As in the matching pennies game, we observe that Y's strategy always converges to its equilibrium independent of X's and Y's initial strategies.

\paragraph{Continuous equilibrium:} Second, Fig.~\ref{F04}B shows the case where all the other elements take $0$. In this case, the payoff matrix is not linearly independent in two ways. Indeed, $\sumj a_{j}\bs{u}_{j}=\bs{0}$ for $\bs{a}=(0,1,0,1)$ and $(1,0,1,0)$. Thus, Nash equilibria exist continuously as $\bs{x}^{*}=r_{\sf X}(0,1/2,0,1/2)+(1-r_{\sf X})(1/2,0,1/2,0)$ and $\bs{y}^{*}=r_{\sf Y}(0,1/2,0,1/2)+(1-r_{\sf Y})(1/2,0,1/2,0)$ for all $0\le r_{\sf X}\le 1$ and $0\le r_{\sf Y}\le 1$. Even in such continuous equilibria, we observe that Y's strategy converges to one of the equilibria depending on X's and Y's initial strategies.

\paragraph{Boundary equilibrium:} Third, Fig.~\ref{F04}C shows the case where the other elements take $0$ in principle except for $u_{11}=-1$. In this case, the payoff matrix is not linearly independent in one way. Indeed, $\sumj a_{j}\bs{u}_{j}=\bs{0}$ for $\bs{a}=(0,1,0,1)$. The only Nash equilibrium exists on the boundary of strategy spaces, $\bs{x}^{*}=(0,1/2,0,1/2)$, and $\bs{y}^{*}=(1/2,0,1/2,0)$. As far as we observe our experiments, Y's strategy fails to converge the equilibrium when the payoff matrix is not linearly independent and is equipped with only the boundary Nash equilibrium. Nevertheless, we observe that X's strategy converges to the equilibrium instead of Y's. We find no case that neither X's nor Y's strategy converges to the equilibrium.


\section{Conclusion}
This study considered the simplest situation of memory asymmetry between two players; only player X memorizes the other's previous action, while player Y cannot. We formulated their learning dynamics based on the replicator dynamics. Although the existence of memory complicates the dynamics, we captured the global behavior of the learning dynamics by introducing two new quantities. One is the conditional-sum divergence, which is an extension of the previous divergence to the case of reactive strategies. We proved that when X exploits Y, this conditional-sum divergence becomes smaller, meaning that their strategies converge to the Nash equilibrium. The other is a family of Lyapunov functions, meaning X's exploitability to Y. We proved that these Lyapunov functions monotonically increase, meaning that X learns to exploit Y with time. As a valid application of the combination of these two quantities, we suggested the global convergence to the Nash equilibrium. Theoretically, we proved the global convergence in the matching pennies game, the simplest game equipped with an interior Nash equilibrium. Our experiments further support that global convergence occurs in coupled matching pennies games, which can have various types of Nash equilibrium structures, such as interior equilibrium, continuous equilibrium, and boundary equilibrium. It is still a conjecture whether the learning dynamics with memory asymmetry converge to the Nash equilibrium in general zero-sum games. This study provides novel and valid indicators to analyze dynamics in learning in games with memories.

\begin{acks}
K. Ariu is supported by JSPS KAKENHI Grant No.~23K19986.
\end{acks}



\balance


\newpage

\onecolumn

\appendix

\begin{center}
{\Large\bf Appendix}
\end{center}

\section{Proofs}
\subsection{Proof of Thm.~\ref{thm_D}} \label{proof_thm_D}
\begin{proof}
The dynamics of $D(\bs{X},\bs{y})$ are calculated as
\begin{align}
    \dot{D}(\bs{X},\bs{y})=&-\sumi\sumj x_{i}^{*}\frac{\dot{x}_{i|j}}{x_{i|j}}-\sumj y_{j}^{*}\frac{\dot{y}_{j}}{y_{j}}
    \nonumber\\
    =&-\sumi\sumj x_{i}^{*}\left(\frac{\rmd u^{\st}}{\rmd x_{i|j}}-\sumi x_{i|j}\frac{\rmd u^{\st}}{\rmd x_{i|j}}\right)-\sumj y_{j}^{*}\left(\frac{\rmd u^{\st}}{\rmd y_{j}}-\sumj y_j\frac{\rmd u^{\st}}{\rmd y_{j}}\right)
    \nonumber\\
    =&-\sumi\sumj x_{i}^{*}y_{j}\left(\sumjp u_{ij'}y_{j'}-\sumi x_{i|j}\sumjp u_{ij'}y_{j'}\right)+\sumj y_{j}^{*}\left(\sumi u_{ij}x_{i}^{\st}-\sumj y_{j}\sumi u_{ij}x_{i}^{\st}\right)
    \nonumber\\
    &+\sumj y_{j}^{*}\left(\sumi x_{i|j}\sumjp u_{ij'}y_{j'}-\sumj y_{j}\sumi x_{i|j}\sumjp u_{ij'}y_{j'}\right)
    \nonumber\\
    =&-\underbrace{\sumi\sumj\sumjp x_{i}^{*}y_{j}u_{ij'}y_{j'}}_{=:{\rm (A)}=u^*}+\underbrace{\sumi\sumj\sumjp x_{i|j}y_{j}u_{ij'}y_{j'}}_{=:{\rm (B)}=u^{\st}}+\underbrace{\sumi\sumj y_{j}^{*}u_{ij}x_{i}^{\st}}_{=:{\rm (C)}=u^*}-\underbrace{\sumi\sumj y_{j}u_{ij}x_{i}^{\st}}_{=:{\rm (D)}=u^{\st}}
    \nonumber\\
    &+\underbrace{\sumi\sumj\sumjp y_{j}^{*}x_{i|j}u_{ij'}y_{j'}-\sumi\sumj\sumjp y_{j}x_{i|j}u_{ij'}y_{j'}}_{=:{\rm (E)}=-\rmd\bs{y}^{T}\bs{X}^{\rm T}\bs{U}\rmd\bs{y}}
    \nonumber\\
    =&-u^*+u^{\st}+u^*-u^{\st}-\rmd\bs{y}^{T}\bs{X}^{\rm T}\bs{U}\rmd\bs{y}
    \nonumber\\
    =&-\rmd\bs{y}^{T}\bs{X}^{\rm T}\bs{U}\rmd\bs{y}
    \nonumber\\
    (=&:D^{\dagger}(\bs{X};\rmd\bs{y})).
\end{align}
Here, we calculated (A)-(D) as
\begin{align}
    {\rm (A)}&=\sumi\sumj\sumjp x_{i}^{*}y_{j}u_{ij'}y_{j'}
    \nonumber\\
    &=\sumj\sumjp y_{j}y_{j'}(\sumi x_{i}^{*}u_{ij'})
    \nonumber\\
    &=\sumj\sumjp y_{j}y_{j'} u^{*}
    \nonumber\\
    &=u^{*}, \\
    {\rm (B)}&=\sumi\sumj\sumjp x_{i|j}y_{j}u_{ij'}y_{j'}
    \nonumber\\
    &=\sumi\sumjp (\sumj x_{i|j}y_{j})u_{ij'}y_{j'}
    \nonumber\\
    &=\sumi\sumjp x_{i}^{\st}u_{ij'}y_{j'}
    \nonumber\\
    &=u^{\st}, \\
    {\rm (C)}&=\sumi\sumj y_{j}^{*}u_{ij}x_{i}^{\st}
    \nonumber\\
    &=\sumi(\sumj y_{j}^{*}u_{ij})x_{i}^{\st}
    \nonumber\\
    &=u^{*}\sumi x_{i}^{\st}
    \nonumber\\
    &=u^{*}, \\
    {\rm (D)}&=\sumi\sumj y_{j}u_{ij}x_{i}^{\st}
    \nonumber\\
    &=u^{\st}, \\
    {\rm (E)}&=\sumi\sumj\sumjp y_{j}^{*}x_{i|j}u_{ij'}y_{j'}-\sumi\sumj\sumjp y_{j}x_{i|j}u_{ij'}y_{j'}
    \nonumber\\
    &=-\sumi\sumj\sumjp \underbrace{(y_{j}-y_{j}^{*})}_{=:\rmd y_{j}}x_{i|j}u_{ij'}\underbrace{y_{j'}}_{=\rmd y_{j'}+y_{j'}^{*}}
    \nonumber\\
    &=-\sumi\sumj\sumjp \rmd y_{j}x_{i|j}u_{ij'}\rmd y_{j'}-\sumi\sumj\sumjp \rmd y_{j}x_{i|j}u_{ij'}y_{j'}^{*}
    \nonumber\\
    &=-\sumi\sumj\sumjp \rmd y_{j}x_{i|j}u_{ij'}\rmd y_{j'}-u^{*}\sumi\sumj y_{j}x_{i|j}
    \nonumber\\
    &=-\sumi\sumj\sumjp \rmd y_{j}x_{i|j}u_{ij'}\rmd y_{j'}
    \nonumber\\
    &=-\rmd\bs{y}^{\rm T}\bs{X}^{\rm T}\bs{U}\rmd\bs{y}.
\end{align}
Here, note that $\rmd\bs{y}=\bs{y}-\bs{y}^{*}$ is zero-sum vector. Thus, if $\bs{X}^{\rm T}\bs{U}$ is positive definite for zero-sum vectors, $D^{\dagger}(\bs{X};\rmd\bs{y})\le 0$ holds independent of $\rmd\bs{y}$, meaning that $D(\bs{X},\bs{y})$ monotonically decreases with time.
\end{proof}

\subsection{Proof of Thm.~\ref{thm_H}} \label{proof_thm_H}
\begin{proof}
By Eqs.~\eqref{dotx} and~\eqref{gradx}, we obtain
\begin{align}
    \dot{H}(\bs{X},\bs{y})=&\sumi\sumip\sumj\delta_{i}u_{ij}\frac{\dot{x}_{i'|j}}{x_{i'|j}}\delta_{i'}
    \nonumber\\
    =&\sumi\sumip\sumj\delta_{i}u_{ij}\left(\frac{\rmd u^{\st}}{\rmd x_{i'|j}}-\sumipp x_{i''|j}\frac{\rmd u^{\st}}{\rmd x_{i''|j}}\right)\delta_{i'}
    \nonumber\\
    =&\sumi\sumip\sumj\delta_{i}u_{ij}\frac{\rmd u^{\st}}{\rmd x_{i'|j}}\delta_{i'}-\sumi\sumj \delta_{i}u_{ij}\sumipp x_{i''|j}\frac{\rmd u^{\st}}{\rmd x_{i''|j}}\underbrace{\sumip\delta_{i'}}_{=0}
    \nonumber\\
    =&\sumi\sumip\sumj\delta_{i}u_{ij}\frac{\rmd u^{\st}}{\rmd x_{i'|j}}\delta_{i'}
    \nonumber\\
    =&\sumi\sumip\sumj\sumjp\delta_{i}u_{ij}y_{j}\delta_{i'}u_{i'j'}y_{j'}
    \nonumber\\
    =&(\sumi\sumj \delta_{i}u_{ij}y_{j})^2
    \nonumber\\
    =&(\sumi\sumj \delta_{i}u_{ij}\rmd y_{j}+\underbrace{\sumi\sumj \delta_{i}u_{ij}y_{j}^{*}}_{=u^{*}\sumi \delta_{i}=0})^2
    \nonumber\\
    =&(\bs{\delta}^{\rm T}\bs{U}\rmd\bs{y})^2
    \nonumber\\
    (=&:H^{\dagger}(\bs{y};\bs{\delta})).
\end{align}
This equation shows that $H^{\dagger}(\bs{y};\bs{\delta})\ge 0$ always holds.

Furthermore, let us prove that $H^{\dagger}(\bs{y};\bs{\delta})=0$ for all $\bs{\delta}$ is equivalent to $\rmd\bs{y}=\bs{0}\Leftrightarrow \bs{y}=\bs{y}^{*}$. Here, we substitute $\bs{\delta}=\bs{x}-\bs{x}^{*}$ for $\bs{x}\in\Delta^{m_{\sf X}-1}$ (but no problem occurs even if we generally consider $\bs{x}\in\mathbb{R}^{m_{\sf X}}$ such that $\sumi x_{i}=1$), and then
\begin{align}
    \bs{\delta}^{\rm T}\bs{U}\rmd\bs{y}&=(\bs{x}-\bs{x}^{*})^{\rm T}\bs{U}(\bs{y}-\bs{y}^{*}) \\
    &=\underbrace{\bs{x}^{\rm T}\bs{U}\bs{y}}_{=u^{\st}(\bs{x},\bs{y})}-\underbrace{\bs{x}^{\rm T}\bs{U}\bs{y}^{*}}_{=u^{*}}-\underbrace{\bs{x}^{*\rm T}\bs{U}\bs{y}}_{=u^{*}}+\underbrace{\bs{x}^{*\rm T}\bs{U}\bs{y}^{*}}_{=u^{*}} \\
    &=u^{\st}(\bs{x},\bs{y})-u^{*}.
\end{align}
We obtain
\begin{align}
    H^{\dagger}(\bs{y};\bs{x}-\bs{x}^{*})=|u^{\st}(\bs{x},\bs{y})-u^{*}|^{2}.
\end{align}
We use the proof by contradiction. Assume $\rmd\hat{\bs{y}}\neq\bs{0}\Leftrightarrow \hat{\bs{y}}\neq\bs{y}^{*}$ such that $u^{\st}(\bs{x},\hat{\bs{y}})=u^{*}$ for all $\bs{x}$, and then we obtain for $\bs{x}=\bs{x}^{*}$
\begin{align}
    &u^{\st}(\bs{x}^{*},\hat{\bs{y}})=u^{*}=\max_{\bs{y}}u^{\st}(\bs{x}^{*},\bs{y}) \\
    &\Leftrightarrow \hat{\bs{y}}=\arg\max_{\bs{y}}u^{\st}(\bs{x}^{*},\bs{y})=\bs{y}^{*},
\end{align}
which contradicts $\hat{\bs{y}}\neq\bs{y}^{*}$. We proved $H^{\dagger}(\bs{y};\bs{\delta})=0$ for all $\bs{\delta}$ is equivalent to $\rmd\bs{y}=\bs{0}\Leftrightarrow \bs{y}=\bs{y}^{*}$.
\end{proof}

\subsection{Proof of Cor.~\ref{cor_global}} \label{proof_cor_global}
\begin{proof}
Now, by using the special properties of two-action games, let us simplify the notation of $\bs{X}:=\{x_{i|j}\}_{i,j}$, $\bs{y}:=\{y_{j}\}_{j}$, and $\bs{\delta}:=\{\delta_{i}\}_{i}$. From the definition, $\sumi x_{i|j}=1\Leftrightarrow x_{1|j}+x_{2|j}=1$ for all $j\in\{1,2\}$ holds, and thus we simply denote it by $x_{1|j}=:x_{j}$ and $x_{2|j}=1-x_{j}$. In a similar manner, because $y_{1}+y_{2}=1$ holds, we simply denote it by $y_{1}:=y$ and $y_{2}=1-y$. Furthermore, since $\sumi \delta_{i}=0$ holds, we simply denote it by $\delta_{1}=:\delta$ and $\delta_{2}=-\delta$. In two-action games, we can write
\begin{align}
    \bs{\delta}^{\rm T}&=\begin{pmatrix}
        \delta_1 & \delta_2
    \end{pmatrix}=\begin{pmatrix}
        \delta & -\delta
    \end{pmatrix},\\
    \log\bs{X}&=\begin{pmatrix}
        \log x_{1|1} & \log x_{1|2} \\
        \log x_{2|1} & \log x_{2|2} \\
    \end{pmatrix}=\begin{pmatrix}
        \log x_{1} & \log x_{2} \\
        \log (1-x_{1}) & \log (1-x_{2}) \\
    \end{pmatrix}.
\end{align}
Using these notations, we calculate $H(\bs{X};\bs{\delta})$ as
\begin{align}
    H(\bs{X};\bs{\delta})&=\bs{\delta}^{\rm T}\bs{U}\log\bs{X}^{\rm T}\bs{\delta}
    \nonumber\\
    &=\begin{pmatrix}
        \delta & -\delta
    \end{pmatrix}\begin{pmatrix}
        1 & -1 \\
        -1 & 1 \\
    \end{pmatrix}\begin{pmatrix}
        \log x_{1} & \log (1-x_{1}) \\
        \log x_{2} & \log (1-x_{2}) \\
    \end{pmatrix}\begin{pmatrix}
        \delta \\
        -\delta \\
    \end{pmatrix}
    \nonumber\\
    &=2\delta^2\left(\log\frac{x_1}{1-x_1}+\log\frac{x_2}{1-x_2}\right)
    \nonumber\\
    &=2\delta^2(q_1-q_2).
\end{align}
Here, in the final line, we used $q_j:=\log x_j-\log (1-x_j)\Leftrightarrow x_j=\exp(q_j)/(1+\exp(q_j))$. The dynamics of $H(\bs{X};\bs{\delta})$ is also calculated as
\begin{align}
    H^{\dagger}(\bs{y};\bs{\delta})&=(\bs{\delta}^{\rm T}\bs{U}\rmd\bs{y})^2
    \nonumber\\
    &=\left(\begin{pmatrix}
        \delta & -\delta
    \end{pmatrix}\begin{pmatrix}
        1 & -1 \\
        -1 & 1 \\
    \end{pmatrix}\begin{pmatrix}
        y-y^* \\
        -y+y^* \\
    \end{pmatrix}\right)^2
    \nonumber\\
    &=16\delta^2(y-y^*)^2.
\end{align}
This equation shows that for $\delta\neq 0$, $H^{\dagger}(\bs{y};\bs{\delta})=0$ holds if and only if $y=y^*$. Thus, after sufficiently long time passes, $\bs{y}=\bs{y}^*$ holds or $H(\bs{X};\bs{\delta})$ diverges to infinity. If so, $H(\bs{X};\bs{\delta})>0$ is satisfied and is equivalent to $q_1>q_2$, which is also equivalent to $x_1>x_2$ from the definition. If $x_1>x_2$ holds, $D^{\dagger}(\bs{X};\rmd\bs{y})<0$ is satisfied because we can derive
\begin{align}
    D^{\dagger}(\bs{X};\rmd\bs{y})&=-\rmd\bs{y}^{\rm T}\bs{X}^{\rm T}\bs{U}\rmd\bs{y}
    \nonumber\\
    &=-\begin{pmatrix}
        y-y^{*} & -y+y^{*}
    \end{pmatrix}\begin{pmatrix}
        x_{1} & 1-x_{1} \\
        x_{2} & 1-x_{2} \\
    \end{pmatrix}\begin{pmatrix}
        u_{11} & u_{12} \\
        u_{21} & u_{22} \\
    \end{pmatrix}\begin{pmatrix}
        y-y^{*} \\
        -y+y^{*} \\
    \end{pmatrix}
    \nonumber\\
    &=-(y-y^{*})^2(x_{1}-x_{2})(u_{11}-u_{12}-u_{21}+u_{22}).
\end{align}
Since $D(\bs{X},\bs{y})$ has its lower bound, $D^{\dagger}(\bs{X};\rmd\bs{y})=0\Leftrightarrow y=y^{*}$ holds after sufficiently long time passes. Thus, it was proved that Y's strategy globally converges to its Nash equilibrium strategy.
\end{proof}

\section{Connection with ordinary definiteness} \label{app_definite}
This study considered positive definiteness with some constraint, i.e., for zero-sum vectors. However, this constrained positive definiteness is related to the ordinarily positive definiteness without constraint, introduced as follows (describing ``tilde'' for distinction).

\begin{definition}[Positive definiteness]
A square matrix $\tilde{\bs{M}}$ is positive definite when for all vectors $\tilde{\bs{\delta}}\neq\bs{0}$, $\tilde{\bs{\delta}}\cdot(\tilde{\bs{M}}\tilde{\bs{\delta}})>0$.
\end{definition}

Now, we show how the positive definiteness for zero-sum vectors is tied to the ordinary positive definiteness. First, freely choose and fix $\hat{k}\in\{1,\cdots,m\}$. From any $m\times m$ matrix $\bs{M}=(m_{kk'})_{k,k'}$, we define a $(m-1)\times(m-1)$ matrix $\tilde{\bs{M}}:=(\tilde{m}_{kk'})_{k\neq\hat{k},k'\neq\hat{k}}$ by
\begin{align}
    \tilde{m}_{kk'}&:=m_{kk'}-m_{k\hat{k}}-m_{\hat{k}k'}+m_{\hat{k}\hat{k}}.
\end{align}
Then, the positive definiteness of $\bs{M}$ for zero-sum vectors is equivalent to the positive definiteness of $\tilde{\bs{M}}$ in the ordinary definition as follows.

\begin{theorem}[Equivalence to positive definiteness]
If and only if $m\times m$ matrix $\bs{M}$ is positive definite for zero-sum vectors, $(m-1)\times(m-1)$ matrix $\tilde{\bs{M}}$ is positive definite.
\end{theorem}

\begin{proof}
For any $(m-1)$-dimensional vector $\tilde{\bs{\delta}}=(\delta_{k})_{k\neq\hat{k}}$, we define a $m$-dimensional zero-sum vector $\bs{\delta}$ by
\begin{align}
    \delta_{k}:=\begin{cases}
        -\Sigma_{k\neq\hat{k}} \tilde{\delta}_{k} & (k=\hat{k}) \\
        \tilde{\delta}_{k} & (k\neq\hat{k}) \\
    \end{cases}. \label{delta_tildelta}
\end{align}
Conversely, it is trivially possible to define $\tilde{\bs{\delta}}$ from any $\bs{\delta}$. Thus, there is one-to-one correspondence between the $m$-dimensional zero-sum vector $\bs{\delta}$ and the $(m-1)$-dimensional (general-sum) vector $\tilde{\bs{\delta}}$. Now, we can prove $\bs{\delta}\cdot(\bs{M}\bs{\delta})=\tilde{\bs{\delta}}\cdot(\tilde{\bs{M}}\tilde{\bs{\delta}})$ as follows.
\begin{align}
    \bs{\delta}\cdot(\bs{M}\bs{\delta})&=\sum_{k}\sum_{k'}\delta_{k}m_{kk'}\delta_{k'}
    \nonumber\\
    &=\sum_{k\neq\hat{k}}\sum_{k'\neq\hat{k}}\delta_{k}m_{kk'}\delta_{k'}+\sum_{k\neq\hat{k}}\delta_{k}m_{k\hat{k}}\delta_{\hat{k}}+\sum_{k'\neq\hat{k}}\delta_{\hat{k}}m_{\hat{k}k'}\delta_{k'}+\delta_{\hat{k}}m_{\hat{k}\hat{k}}\delta_{\hat{k}}
    \nonumber\\
    &=\sum_{k\neq\hat{k}}\sum_{k'\neq\hat{k}}\tilde{\delta}_{k}m_{kk'}\tilde{\delta}_{k'}-\sum_{k\neq\hat{k}}\sum_{k'\neq\hat{k}}\tilde{\delta}_{k}m_{k\hat{k}}\tilde{\delta}_{k'}-\sum_{k\neq\hat{k}}\sum_{k'\neq\hat{k}}\tilde{\delta}_{k}m_{\hat{k}k'}\tilde{\delta}_{k'}+\sum_{k\neq\hat{k}}\sum_{k'\neq\hat{k}}\tilde{\delta}_{k}m_{\hat{k}\hat{k}}\tilde{\delta}_{k'}
    \nonumber\\
    &=\sum_{k\neq\hat{k}}\sum_{k'\neq\hat{k}}\tilde{\delta}_{k}(m_{kk'}-m_{k\hat{k}}-m_{\hat{k}k'}+m_{\hat{k}\hat{k}})\tilde{\delta}_{k'}
    \nonumber\\
    &=\sum_{k\neq\hat{k}}\sum_{k'\neq\hat{k}}\tilde{\delta}_{k}\tilde{m}_{kk'}\tilde{\delta}_{k'}
    \nonumber\\
    &=\tilde{\bs{\delta}}\cdot(\tilde{\bs{M}}\tilde{\bs{\delta}}).
\end{align}
Here, in the third equal sign, we used Eq.~\eqref{delta_tildelta}. Thus, if and only if $\bs{M}$ is positive definite for zero-sum vectors (i.e., $\bs{\delta}\cdot(\bs{M}\bs{\delta})>0$), $\tilde{\bs{M}}$ is positive definite (i.e., $\tilde{\bs{\delta}}\cdot(\tilde{\bs{M}}\tilde{\bs{\delta}})>0$).
\end{proof}

\section{Connection to the classical divergence} \label{app_divergence}
This section is dedicated to understanding the connection between our conditional-sum divergence $D(\bs{X},\bs{y})$ and the classical divergence $D_{\rm c}(\bs{x},\bs{y})$. Here, recall that X's reactive strategy consists of multiple vectors $\bs{X}=(\bs{x}_{j})_{1\le j\le m_{\sf Y}}$, where $\bs{x}_{j}$ is the vector of the probability distribution of its action choice when the other chose $j$-th action in the previous round. On the other hand, in the classical divergence, X's mixed strategy $\bs{x}$ is the single vector of the probability distribution of its action choice independent of the other's previous choice. Thus, if the reactive strategy of $\bs{X}$ satisfies $\bs{x}_{j}=\bs{x}$ for all $j$, it corresponds to the mixed strategy of $\bs{x}$. Then, $\dot{D}(\bs{X},\bs{y})=\dot{D}_{\rm c}(\bs{x},\bs{y})$ holds, as the following theorem shows.

\begin{theorem}[Connection with the classical total divergence]
When $\bs{X}$ satisfies $\bs{x}_{j}=\bs{x}$ for all $j$, $\dot{D}(\bs{X},\bs{y})=\dot{D}_{\rm c}(\bs{x},\bs{y})=0$ holds.
\end{theorem}

\begin{proof}
Let us prove that $\dot{D}=\sumj \dot{D}_{\rm KL}(\bs{x}^{*}\|\bs{x}_{j})+\dot{D}_{\rm KL}(\bs{y}^{*}\|\bs{y})$ and $\dot{D}_{\rm c}=\dot{D}_{\rm KL}(\bs{x}^{*}\|\bs{x})+\dot{D}_{\rm KL}(\bs{y}^{*}\|\bs{y})$ are equal. First, the time evolution of Y's divergence, i.e., $\dot{D}_{\rm KL}(\bs{y}^{*}\|\bs{y})$, is trivially equal between these two equations. Thus, we prove $\sumj \dot{D}_{\rm KL}(\bs{x}^{*}\|\bs{x}_{j})=\dot{D}_{\rm KL}(\bs{x}^{*}\|\bs{x})$ below.
\begin{align}
    \sumj\dot{D}_{\rm KL}(\bs{x}^{*}\|\bs{x}_{j})&=-\sumj\sumi x_{i}^{*}\frac{\dot{x}_{i|j}}{x_{i|j}}
    \nonumber\\
    &=-\sumj\sumi x_{i}^{*}y_{j}(\sumjp u_{ij'}y_{j'}-\sumi \underbrace{x_{i|j}}_{=x_{i}}\sumjp u_{ij'}y_{j'})
    \nonumber\\
    &=-\sumi x_{i}^{*}(\sumjp u_{ij'}y_{j'}-\sumi x_{i}\sumjp u_{ij'}y_{j'})
    \nonumber\\
    &=\dot{D}_{\rm KL}(\bs{x}^{*}\|\bs{x}).
\end{align}
In the final equal sign, we used the replicator dynamics of X's strategy for games without memories as
\begin{align}
    \dot{x}_{i}&=x_{i}\left(\frac{\rmd u_{\st}(\bs{x},\bs{y})}{\rmd x_i}-\sumi x_i\frac{\rmd u_{\st}(\bs{x},\bs{y})}{\rmd x_i}\right)
    \nonumber\\
    &=x_{i}(\sumj u_{ij}y_{j}-\sumi x_{i}\sumj u_{ij}y_{j})
\end{align}
Thus, we proved $\dot{D}=\dot{D}_{\rm c}$.
\end{proof}

In this proof, we proved $\dot{D}(\bs{X},\bs{y})=\dot{D}_{\rm c}(\bs{x},\bs{y})$ separately for each divergence of X's strategy and Y's strategy.

Furthermore, note that $\dot{D}(\bs{X},\bs{y})=\dot{D}_{\rm c}(\bs{x},\bs{y})$ holds only for a moment. Because the learning dynamics give $\dot{\bs{x}}_{j}\neq \dot{\bs{x}}$, $\bs{x}_{j}$ gradually differs from $\bs{x}$. In conclusion, $\dot{D}(\bs{X},\bs{y})=\dot{D}_{\rm c}(\bs{x},\bs{y})$ becomes unsatisfied immediately. (However, $\dot{D}_{\rm c}(\bs{x},\bs{y})=0$ continues to hold.)

\section{Extension to other learning algorithms} \label{app_extension}
In the main text, we focus only on the replicator dynamics. However, our findings are also applicable to the gradient descent-ascent algorithm. In the interior of the strategy space, the learning dynamics of the gradient descent-ascent are described as
\begin{align}
    \dot{x}_{i|j}=+\left(\frac{\rmd u^{\st}}{\rmd x_{i|j}}-\sumi \frac{1}{m_{\sf X}}\frac{\rmd u^{\st}}{\rmd x_{i|j}}\right),\quad \dot{y}_{j}=-\left(\frac{\rmd u^{\st}}{\rmd y_{j}}-\sumj \frac{1}{m_{\sf Y}}\frac{\rmd u^{\st}}{\rmd y_{j}}\right),
\end{align}
which correspond to Eqs.~\eqref{dotx} and~\eqref{doty}, respectively.

To capture the global behavior in the learning dynamics based on the gradient descent-ascent, we have to make small changes in $D(\bs{X},\bs{y})$ and $H(\bs{X};\bs{\delta})$ as
\begin{align}
    D(\bs{X},\bs{y})&:=\sumj D_{2}(\bs{x}^{*},\bs{x}_{j})+D_{2}(\bs{y}^{*},\bs{y}),\quad D_{2}(\bs{p}^{*},\bs{p}):=\frac{1}{2}\|\bs{p}-\bs{p}^{*}\|^2, \\
    H(\bs{X};\bs{\delta})&:=\bs{\delta}^{\rm T}\bs{U}\bs{X}^{\rm T}\bs{\delta},
\end{align}
which correspond to Eqs.~\eqref{D_function} and~\eqref{H_function}, respectively.

This $D(\bs{X},\bs{y})$ satisfies Thm.~\ref{thm_D} as follows.
\begin{align}
    \dot{D}(\bs{X},\bs{y})&=\sumi\sumj\dot{x}_{i|j}(x_{i|j}-x_{i}^{*})+\sumj\dot{y}_{j}(y_{j}-y_{j}^{*})
    \nonumber\\
    &=\sumi\sumj (x_{i|j}-x_{i}^{*})\left(\frac{\rmd u^{\st}}{\rmd x_{i|j}}+-\sumi\frac{1}{m_{\sf X}}\frac{\rmd u^{\st}}{\rmd x_{i|j}}\right)-\sumj (y_{j}-y_{j}^{*})\left(\frac{\rmd u^{\st}}{\rmd  y_{j}}-\sumj\frac{1}{m_{\sf Y}}\frac{\rmd u^{\st}}{\rmd y_{j}}\right)
    \nonumber\\
    &=\sumi\sumj (x_{i|j}-x_{i}^{*})\frac{\rmd u^{\st}}{\rmd x_{i|j}}-\sumj (y_{j}-y_{j}^{*})\frac{\rmd u^{\st}}{\rmd y_{j}}
    \nonumber\\
    &=\underbrace{\sumi\sumj (x_{i|j}-x_{i}^{*})y_{j}\sumjp u_{ij'}y_{j'}}_{=({\rm B})-({\rm A})=u^{\st}-u^{*}}-\underbrace{\sumj (y_{j}-y_{j}^{*})\sumi u_{ij}x_{i}^{\st}}_{=:({\rm D})-({\rm C})=u^{\st}-u^{*}}-\underbrace{\sumj (y_{j}-y_{j}^{*})\sumi x_{i|j}\sumjp u_{ij'}y_{j'}}_{=({\rm E})=\rmd\bs{y}^{T}\bs{X}^{\rm T}\bs{U}\rmd\bs{y}}
    \nonumber\\
    &=-\rmd\bs{y}^{T}\bs{X}^{\rm T}\bs{U}\rmd\bs{y}
    \nonumber\\
    &=D^{\dagger}(\bs{X};\rmd\bs{y}).
\end{align}
Here, we used the terms $({\rm A})$-$({\rm E})$ in Appendix~\ref{proof_thm_D}.

Furthermore, this $H(\bs{X};\bs{\delta})$ also satisfies Thm.~\ref{thm_H} as follows.
\begin{align}
    \dot{H}(\bs{X};\bs{\delta})&=\sumi\sumip\sumj \delta_{i}u_{ij}\dot{x}_{i'|j}\delta_{i'}
    \nonumber\\
    &=\sumi\sumip\sumj \delta_{i}u_{ij}\left(\frac{\rmd u^{\st}}{\rmd x_{i'|j}}+-\sumipp\frac{1}{m_{\sf X}}\frac{\rmd u^{\st}}{\rmd x_{i''|j}}\right)\delta_{i'}
    \nonumber\\
    &=\sumi\sumip\sumj \delta_{i}u_{ij}\frac{\rmd u^{\st}}{\rmd x_{i'|j}}\delta_{i'}+\sumi\sumj\delta_{i}u_{ij}\sumipp\frac{1}{m_{\sf X}}\frac{\rmd u^{\st}}{\rmd x_{i''|j}}\underbrace{\sumip\delta_{i'}}_{=0}
    \nonumber\\
    &=\sumi\sumip\sumj \delta_{i}u_{ij}\frac{\rmd u^{\st}}{\rmd x_{i'|j}}\delta_{i'}
    \nonumber\\
    &=H^{\dagger}(\bs{y};\bs{\delta}).
\end{align}
In the final line, we used the same calculation as that of Appendix~\ref{proof_thm_H}.

Both the replicator dynamics and the gradient descent-ascent are included in the Follow the Regularized Leaders. A promising future study is to capture the global behavior in learning dynamics based on general Follow the Regularized Leaders algorithms with memories.


\end{document}